\documentclass[10pt,conference]{IEEEtran}\IEEEoverridecommandlockouts
\usepackage{epsfig,graphics,subfigure,psfrag,amsmath,cases}
\usepackage{latexsym,amssymb,amsmath,epsfig,subfigure,algorithm}
\usepackage{algorithmic}
\usepackage{color}
\usepackage{url}
\usepackage{scrtime}
\usepackage{array}

\author{\authorblockN{ Derrick Wing Kwan Ng\authorrefmark{1}, Robert Schober\authorrefmark{2}, and Hussein Alnuweiri\authorrefmark{3}
\thanks{\authorrefmark{2}The author is also with the University of British Columbia, Vancouver, Canada. This research was supported by the Qatar National Research Fund (QNRF), under project NPRP 5-401-2-161. }}
Institute for Digital Communications, Friedrich-Alexander-University Erlangen-N\"urnberg (FAU), Germany\authorrefmark{1}\authorrefmark{2}\\
Texas A\&M University at Qatar,  Qatar\authorrefmark{3}\\
Email: kwan@lnt.de,  schober@lnt.de, hussein.alnuweiri@qatar.tamu.edu \vspace*{-0.3cm}

}

\title{\vspace*{-0.5cm}Power Efficient MISO Beamforming for Secure  Layered Transmission}

\date{\thistime,\,\today}

\newtheorem{proposition}{Proposition}

\DeclareMathOperator{\Tr}{Tr}
\DeclareMathOperator{\Rank}{Rank}

\newtheorem{Remark}{Remark}
\DeclareMathOperator{\maxo}{maximize}
\DeclareMathOperator{\mino}{minimize}
\newcommand{\textoverline}[1]{$\overline{\mbox{#1}}$}
\newcommand{\abs}[1]{\lvert#1\rvert}
\newcommand{\norm}[1]{\lVert#1\rVert}
\newcolumntype{L}{>{\centering\arraybackslash}m{3cm}}
\begin{document}

\maketitle

\begin{abstract}
This paper studies secure layered video transmission in a multiuser multiple-input single-output (MISO) beamforming downlink communication system.  The power
allocation algorithm design is formulated as a non-convex
optimization problem for  minimizing the total transmit power while
 guaranteeing a minimum  received  signal-to-interference-plus-noise ratio (SINR) at the desired receiver.
 In particular, the proposed problem formulation takes into account the  self-protecting architecture of layered
 transmission and artificial noise generation  to prevent potential information eavesdropping. A semi-definite programming (SDP) relaxation  based power allocation algorithm is proposed to obtain an upper bound solution. A sufficient condition for the global optimal solution is examined to reveal the tightness of the upper bound solution. Subsequently, two suboptimal power allocation schemes with low computational complexity are proposed for enabling secure layered video transmission. Simulation results demonstrate  significant transmit power savings achieved by the proposed algorithms
 and layered transmission compared to the baseline schemes.

\end{abstract}
\renewcommand{\baselinestretch}{0.94}
\large\normalsize

\section{Introduction}
\label{sect1}
The explosive growth in  high data  rate real time multimedia services in wireless communication networks  has
led to a heavy demand for energy and bandwidth. Multiple-input multiple-output (MIMO) technology has emerged as one of the most prominent solutions in fulfilling this challenging demand, due to its inherited  extra degrees of freedom for resource allocation \cite{book:david_wirelss_com}--\nocite{JR:TWC_large_antennas,JR:MIMO_layered,JR:MIMO_layered2}\cite{CN:MISO_layer}. However, the hardware complexity of multiple antenna receivers limits the deployment of such technology in practice, especially for portable
devices. As an alternative, multiuser MIMO  has been proposed   where a
transmitter equipped with multiple  antennas services multiple
single-antenna users.

Furthermore, for video streaming, layered transmission  has been implemented in some existing video standards such as  H.264/Moving Picture Experts Group (MPEG)-4 scalable video coding (SVC)  \cite{JR:Video_layers,JR:Video_layers2}. Specifically, a video signal is encoded into a hierarchy of  multiple layers with unequal importance, namely one
 base layer and several enhancement layers. The base layer contains the essential information of the video with minimum video quality. The information embedded in each enhancement layer is used to successively  refine the description of the pervious layers.  The structure of layered transmission  facilities the implementation of unequal error protection. In fact,  the transmitter can achieve a better resource utilization  by  allocating different amount of
powers  to different information layers according to their importance to the video quality.

Also, since the broadcast nature of wireless communication channels makes them vulnerable to eavesdropping, there is an emerging need for guaranteeing secure wireless video communication. For instance, misbehaving legitimate users of a communication system may attempt to use the   high definition video  service without paying by overhearing the video signal. Although cryptographic encryption algorithms are commonly implemented in the application layer for providing secure communication, the associated  security key distribution and management can be problematic or infeasible in wireless networks.  As a result,  physical (PHY) layer
security \cite{Report:Wire_tap}--\nocite{CN:Multicast,JR:Artifical_Noise1,JR:Kwan_physical_layer}\cite{JR:secrecy_beamforming} has been proposed as a complement to the traditional methods for improving wireless transmission security. The merit of PHY layer security lies in the guaranteed  perfect secrecy of communication by exploiting the physical characteristics of the wireless  communication channel. In his seminal work on PHY layer security \cite{Report:Wire_tap}, Wyner showed
 that a non-zero secrecy capacity,
defined as the maximum transmission rate at which an eavesdropper
is unable to extract any information from the received signal, can be achieved if
the desired receiver enjoys better channel conditions than the
eavesdropper. A considerable amount of research \cite{CN:Multicast}--\cite{JR:Kwan_physical_layer} has been devoted
to exploiting multiple antennas for  providing  communication secrecy. In \cite{CN:Multicast}, transmit beamforming was proposed to maximize the secrecy capacity in a multiple-input single-output (MISO) communication system.
 In \cite{JR:Artifical_Noise1} and  \cite{JR:Kwan_physical_layer}, artificial noise generation
was proposed for multiple antenna communication systems to cripple the interception capabilities  of eavesdroppers. In particular, \cite{JR:Artifical_Noise1} and  \cite{JR:Kwan_physical_layer} focused on the
maximization of the ergodic secrecy capacity and the outage secrecy capacity for different system settings, respectively.
 In \cite{JR:secrecy_beamforming}, a power allocation algorithm was proposed to maximize the system energy efficiency  while providing delay-constrained secure communication service. However, the results of \cite{CN:Multicast}--\cite{JR:secrecy_beamforming}
are based on the assumption of  single layer transmission  and may not be applicable to multimedia video transmission. Besides, the
layered information architecture  of video signals has a
\emph{self-protecting structure} which provides a certain robustness against  eavesdropping. However, exploiting the  layered transmission for facilitating PHY layer security in video communication  has not been considered in the literature \cite{JR:MIMO_layered}--\cite{JR:Video_layers2}.

Motivated by the aforementioned observations, we
formulate the power allocation algorithm design for secure layered video transmission in multiuser MISO beamforming systems
 as a non-convex optimization problem. We propose a semi-definite programming (SDP) approach for designing a power allocation algorithm which obtains an upper bound solution for the considered problem. Besides, we use the upper bound solution  as a building block for the design of two suboptimal beamforming schemes with low computational complexity and near optimal performance.

 \begin{figure*}[t]
 \centering \vspace*{-3mm}
\includegraphics[width=4.5in]{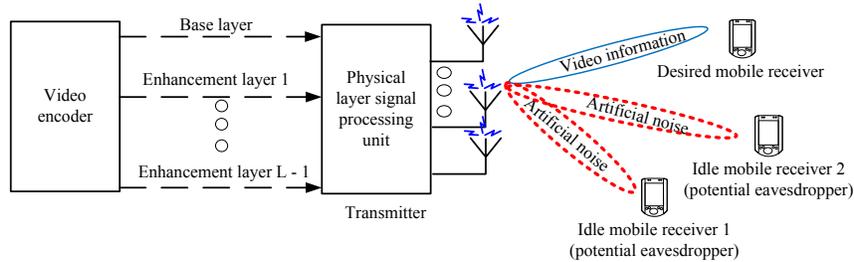} \vspace*{-1mm}
 \caption{Layered transmission system model for $K=3$ mobile video receivers.} \label{fig:system_model} \vspace*{-4mm}
\end{figure*}

\section{System Model}
\label{sect:OFDMA_AF_network_model}
In this section,  we  present the adopted models for  multiuser downlink communication and  layered video encoding.

\subsection{Channel  Model}
A downlink communication system is considered which consists of a transmitter and $K$  receivers. The transmitter is
equipped with $N_\mathrm{T}$ transmit antennas while the receivers are single antenna devices, cf. Figure \ref{fig:system_model}.
The transmitter conveys   video information to a given video service  subscriber (receiver) while the remaining $K-1$ receivers are idle.
However, the transmitted video signals  may be overheard by the idle receivers. In practice, it is possible that the
idle receivers are malicious and eavesdrop the video information of the other subscribers, e.g., a paid multimedia video service, by overhearing the  video
 signal transmitted by the transmitter. In other words, the idle receivers are potential eavesdroppers which should be taken into account when delivering secure video information to the desired user.
 We focus on a frequency flat slow time varying fading channel in a time division duplexing (TDD) system. The
downlink channel gains of all receivers can be accurately estimated based on the uplink pilot sequences contained in the handshaking signals via channel reciprocity. The downlink
received signals at the desired receiver and the $K-1$ idle receivers are given by, respectively,
\begin{eqnarray}
y&=&\mathbf{h}^{H} \mathbf{x}+z_\mathrm{s},\\
y_{\mathrm{E}_k}&=&\mathbf{g}_{k}^{H} \mathbf{x} +z_{\mathrm{s},k},\,\,  \forall k=\{1,\ldots,K-1\},
\end{eqnarray}
where $\mathbf{x}\in\mathbb{C}^{ N_\mathrm{T} \times 1}$ denotes the transmitted symbol vector and $\mathbb{C}^{N\times M}$ denotes the space of $N\times M$ matrices with complex entries.
$\mathbf{h}^{H}\in\mathbb{C}^{1\times N_\mathrm{T}}$ is the channel
vector between the transmitter and the desired receiver
  and
$\mathbf{g}_{k}^{H} \in\mathbb{C}^{1\times N_\mathrm{T}}$ is the channel
vector between the transmitter and idle receiver (potential eavesdropper) $k$. $(\cdot)^H$ denotes the conjugate transpose of  an input matrix.  Variables
 $\mathbf{h}$ and $\mathbf{g}_{k}$ capture the effects of the multipath fading and path loss
of the associated channels.
  The noise terms $z_\mathrm{s}$ and $z_{\mathrm{s},k}$ include the joint effect of
  thermal noise and signal processing noise in the desired receiver and idle receiver $k$, respectively. They are modeled as  additive white Gaussian noises   with zero mean and variance\footnote{We assume that the signal processing and thermal noise characteristics of all receivers are identical due to a similar hardware architecture.} $\sigma_{\mathrm{s}}^2$,  cf. Figure \ref{fig:system_model}.

\subsection{Video Encoding and Artificial Noise Generation}
We assume that the video source is encoded  into $L$ layers via scalable video coding. Without loss of generality, the
video information can be represented as $\mathbf{S}=\big[s_1,s_2,\ldots,s_L\big],
s_i\in\mathbb{C}, i\in\{1,\ldots,L\}$, where $s_i$ denotes the video information of layer $i$. These $L$ layers consists of one
base layer, i.e., $s_1$, which
can be  decoded independently without utilizing the information from other layers. Specifically, the base layer data includes the
 most essential information of the video and can guarantee a minimum quality of service (QoS).
  The remaining $L-1$ layers, i.e., $s_2,\ldots,s_L$, are  enhancement layers which can be used to successively refine the previous layers. In other words, the enhancement layers cannot be decoded independently;  if the decoding of a layer fails, the information embedded in the following enhancement layers is lost since they are no longer decodable.
In this paper, we consider a fixed rate video source encoder and the data rate of each layer is fixed. This implementation has been supported by
 some standards such as H.264/MPEG-4 SVC  \cite{JR:Video_layers,JR:Video_layers2}.

On the other hand, for guaranteeing secure video transmission
 to the desired receiver, an artificial noise
signal is transmitted along with the information signal to degrade the channels between
the transmitter and the idle receivers (potential eavesdroppers). The transmitter constructs a transmit symbol vector $\mathbf{x}$ as
\begin{eqnarray}
\mathbf{x}=\underbrace{\sum_{i=1}^L \mathbf{w}_i s_i}_{\mbox{desired $L$-layer video signals}}+\underbrace{\mathbf{v}}_{\mbox{artificial noise}},
\end{eqnarray}
where $\mathbf{w}_i\in\mathbb{C}^{N_\mathrm{T}\times 1}$ is the   beamforming vector for the video information signal in layer $i$ dedicated to the desired receiver.
We assume without loss of generally that  ${\cal E}\{\abs{s_i}^2\}=1,\forall i,$ and superposition coding is used to
superimpose the $L$ layers of video information, where ${\cal E}\{\cdot\}$  represents statistical
expectation. $\mathbf{v}\in\mathbb{C}^{N_\mathrm{T}\times 1}$ is the artificial noise vector generated by the transmitter to combat the potential eavesdroppers. In particular,  $\mathbf{v}$ is
modeled as a circularly symmetric complex Gaussian random vector represented as
\begin{eqnarray}
\mathbf{v}\sim {\cal CN}(\mathbf{0}, \mathbf{V}),
\end{eqnarray}
where $\mathbf{V}\in \mathbb{H}^{N_\mathrm{T}},\mathbf{V}\succeq \mathbf{0}$  denotes the covariance matrix of the artificial noise.  Here, $\mathbb{H}^N$   represents the set of all $N$-by-$N$ complex Hermitian matrices and $\mathbf{V}\succeq \mathbf{0}$ indicates that $ \mathbf{V}$ is a  positive semi-definite Hermitian matrix.

\section{Power Allocation Algorithm Design}\label{sect:forumlation}
In this section, we define the adopted quality of service (QoS)
measure for secure communication systems and formulate the power allocation algorithm design as an optimization problem.
We define the following variables for the sake of notational simplicity:
$\mathbf{H}=\mathbf{h}\mathbf{h}^H$ and $\mathbf{G}_k=\mathbf{g}_k\mathbf{g}_k^H, k\in\{1,\ldots,K-1\}$.

\subsection{Channel Capacity}
\label{subsect:Instaneous_Mutual_information}

With successive interference cancellation\footnote{ The corner points of the dominant face
of the multiple-access capacity region can be achieved by superposition coding and a successive interference cancellation receiver  with low computational complexity \cite{book:david_wirelss_com}.
}, the receivers first decode and cancel the lower layers before decoding the higher layers. Besides, the yet to be decoded higher layers are treated as Gaussian noise \cite{JR:interference_model}. Assuming perfect channel state information (CSI) at the
receiver, the capacity (bit/s/Hz) between the transmitter and the desired mobile video receiver of layer $i$
is given by
\begin{eqnarray}\notag
C_i&=&\log_2\Big(1+\Gamma_i\Big)\,\,\,\,
\mbox{and}\,\,\\
\Gamma_i&=&\frac{\abs{\mathbf{h}^H\mathbf{w}_i}^2}
{\sum_{j=i+1}^L\abs{\mathbf{h}^H\mathbf{w}_j}^2+\Tr(\mathbf{H}\mathbf{V})+\sigma_\mathrm{s}^2}\label{eqn:cap} ,
\end{eqnarray}
where $\Gamma_i$ is the received signal-to-interference-plus-noise ratio (SINR) of layer $i$ at the desired receiver.  $\Tr(\cdot)$ and  $\abs{\cdot}$ denote the
 trace of a matrix and the absolute value of a complex scalar, respectively.
On the other hand, the channel capacity between the transmitter and idle receiver (potential eavesdropper) $k$  of layer $i$  is given  by
\begin{eqnarray}\notag
C_{\mathrm{E}_{i,k}}&=&\log_2\Big(1+\Gamma_{\mathrm{E}_{i,k}}\Big)\,\,\,\,
\mbox{and}\,\,\\
\label{eqn:eavesdropper-SINR}
\Gamma_{\mathrm{E}_{i,k}}&=&\frac{\abs{\mathbf{g}_k^H\mathbf{w}_i}^2}
{\sum_{j=i+1}^L \abs{\mathbf{g}_k^H\mathbf{w}_j}^2+\Tr(\mathbf{G}_k\mathbf{V})+\sigma_{\mathrm{s}}^2} , \label{eqn:cap-eavesdropper} \end{eqnarray}
where $\Gamma_{\mathrm{E}_{i,k}}$ is the received SINR at idle receiver $k$.   It can be observed from (\ref{eqn:eavesdropper-SINR}) that layered transmission has a \emph{self-protecting structure}.  Specifically, the higher layer information via the first term in the denominator of (\ref{eqn:eavesdropper-SINR}) has the same effect as the artificial noise signal  $\mathbf{v}$ in protecting the important information encoded in the
lower layers of the video signal. It is expected that by carefully optimizing the beamforming vectors of the higher information layers, a certain level of communication security can be achieved in the low layers.  In the literature, one of the performance metrics for quantifying the notion of communication security in the PHY layer is the secrecy capacity. Specifically,  the maximum secrecy capacity between the transmitter
and the desired receiver on layer $i$ is given by
\begin{eqnarray}\label{eqn:secrecy_cap}
C_{\mathrm{sec}_i}=\Big[C_i - \underset{k\in\{1,\ldots,K-1\}}{\max} C_{\mathrm{E}_{i,k}}\Big]^+,
\end{eqnarray}
where  $[x]^+=\max\{0,x\}$.  $C_{\mathrm{sec}_i}$  quantifies
the maximum achievable data rate at which a transmitter can reliably send a secret
information on layer $i$ to the desired receiver such that the eavesdropper is unable to decode the received signal \cite{Report:Wire_tap}.
\subsection{Optimization Problem Formulation}
\label{sect:cross-Layer_formulation}
The optimal beamforming vector of video information layer $i$, ${\mathbf w}^*_i$, and the optimal artificial noise covariance matrix, ${\mathbf  V}^*$, for minimizing the total radiated power  can be
obtained by solving
\begin{eqnarray}
\label{eqn:cross-layer}&&\hspace*{5mm} \underset{\mathbf{V}\in \mathbb{H}^{N_\mathrm{T}},\mathbf{w}_i
}{\mino} \,\, \sum_{i=1}^L\norm{\mathbf{w}_i}^2+\Tr(\mathbf{V})\nonumber\\
\notag \mbox{s.t.} \hspace*{-1mm}&&\hspace*{-5mm}\mbox{C1: }\notag\frac{\abs{\mathbf{h}^H\mathbf{w}_i}^2}{\sum_{j=i+1}^L \abs{\mathbf{h}^H\mathbf{w}_j}^2+\Tr(\mathbf{H}
\mathbf{V})+\sigma_{\mathrm{s}}^2} \ge \Gamma_{\mathrm{req}_i}, \forall i, \\
\hspace*{-1mm}&&\hspace*{-5mm}\mbox{C2: }\notag\frac{\abs{\mathbf{g}_k^H\mathbf{w}_1}^2}
{\sum_{j=2}^L \abs{\mathbf{g}_k^H\mathbf{w}_j}^2+\Tr(\mathbf{G}_k\mathbf{V})+\sigma_{\mathrm{s}}^2} \le \Gamma_{\mathrm{tol}_k},\forall k,\\
\hspace*{-1mm}&&\hspace*{-5mm}\mbox{C3: }\notag \notag \big[\hspace*{-0.5mm}\sum_{i=1}^{L}\mathbf{w}_i\mathbf{w}^H_i\hspace*{-0.5mm}\big]_{n,n}  \hspace*{-1.5mm}+\hspace*{-1mm}\big[\hspace*{-0.5mm}\mathbf{V}\hspace*{-0.5mm}\big]_{n,n}\hspace*{-1mm}\le P_{\max_n}, \forall n\in\{1,\ldots,N_{\mathrm{T}}\}, \\
\hspace*{-1mm}&&\hspace*{-5mm}\mbox{C4:}\,\, \mathbf{V}\succeq \mathbf{0},
\end{eqnarray}
where $\norm{\cdot}$ in the objective function denotes the Euclidean vector norm. $\Gamma_{\mathrm{req}_i}$ in C1 is the minimum required SINR  for decoding layer $i$ at the desired receiver.
$\Gamma_{\mathrm{tol}_k}$ in C2 denotes the maximum tolerable SINR at idle receiver (potential eavesdropper) $k$ for
decoding the first layer. We note that since layered coding is employed for  encoding the video information,    it is sufficient to protect the first layer of the video for ensuring secure video delivery. In practice,
the transmitter  sets a  sufficiently small value of $\Gamma_{\mathrm{tol}_k}$ such that $\Gamma_{\mathrm{req}_1}\gg\Gamma_{\mathrm{tol}_k}$ holds which  provides an  adequate protection for the first layer. We note that in this paper, we do not directly maximize the secrecy capacity of video delivery. Nevertheless, the proposed problem formulation in (\ref{eqn:cross-layer})
is able to guarantee  a minimum secrecy capacity for layer 1, i.e., $C_{\mathrm{sec}_1}\ge
\log_2(1+\Gamma_{\mathrm{req}_1})-\log_2(1+\underset{k\in\{1,\ldots,K-1\}}{\max}\,\Gamma_{\mathrm{tol}_k})$.
 In C3, $\big[\cdot\big]_{a,b}$ extracts the $(a,b)$-th element of the input matrix. Specifically, C3 limits the maximum transmit power of antenna $n$ to $P_{\max_n}$. C4
 and $\mathbf{V}\in \mathbb{H}^{N_\mathrm{T}}$  are imposed such that $\mathbf{V}$ is a positive semi-definite matrix  to satisfy the physical requirements on covariance matrices.

\section{Solution of the Optimization Problem} \label{sect:solution}
The optimization problem in (\ref{eqn:cross-layer}) is  non-convex due to constraints C1 and C2. In some cases, an exhaustive search is required to obtain the solution of non-convex problems. In order to strike a balance between the optimality of the solution  and the computational complexity of the considered problem, we first recast the problem as a
convex optimization problem by SDP relaxation and obtain a performance upper bound. Then, we propose two computational efficient suboptimal power allocation schemes.
\subsection{Semi-definite Programming Relaxation} \label{sect:solution_dual_decomposition}
For facilitating the SDP relaxation, we define $\mathbf{W}_i=\mathbf{w}_i\mathbf{w}^H_i$ and rewrite problem (\ref{eqn:cross-layer}) in terms of $\mathbf{W}_i$ as
\begin{eqnarray}
\label{eqn:SDP}&&\hspace*{-5mm} \underset{\mathbf{V}\in \mathbb{H}^{N_\mathrm{T}},\mathbf{W}_i\in \mathbb{H}^{N_\mathrm{T}}
}{\mino}\,\, \sum_{i=1}^L \Tr(\mathbf{W}_i)+\Tr(\mathbf{V})\nonumber\\
\notag \mbox{s.t.} \hspace*{-1mm}&&\hspace*{-5mm}\mbox{C1: }\notag\frac{\Tr(\mathbf{H}\mathbf{W}_i)}{\sum_{j=i+1}^L \Tr(\mathbf{H}\mathbf{W}_j)+\Tr(\mathbf{H}
\mathbf{V})+\sigma_{\mathrm{s}}^2} \ge \Gamma_{\mathrm{req}_i}, \forall i, \\
\hspace*{-1mm}&&\hspace*{-5mm}\mbox{C2: }\notag\frac{\Tr(\mathbf{G}_k\mathbf{W}_1)}
{\sum_{j=2}^L \Tr(\mathbf{G}_k\mathbf{W}_j)+\Tr(\mathbf{G}_k\mathbf{V})+\sigma_{\mathrm{s}}^2} \le \Gamma_{\mathrm{tol}_k},\forall k,\\
\hspace*{-1mm}&&\hspace*{-5mm}\mbox{C3: }\notag \notag \Tr\Big(\mathbf{\Psi}_n \big(\mathbf{V}+\sum_{i=1}^L\mathbf{W}_i \big)\Big)\le P_{\max_n}, \forall n, \\
\hspace*{-1mm}&&\hspace*{-5mm}\mbox{C4:}\,\,\mathbf{V}\succeq \mathbf{0},\quad\mbox{C5:}\,\,   \mathbf{W}_i\succeq \mathbf{0}, \forall i, \notag\\
\hspace*{-1mm}&&\hspace*{-5mm}\mbox{C6:}\,\,  \Rank(\mathbf{W}_i)=1, \forall i,
\end{eqnarray}
where $\Rank(\cdot)$ in C6 denotes the rank of the input matrix. $\mathbf{W}_i\succeq \mathbf{0},\forall i\in\{1,\ldots,L\}$,
 $\mathbf{W}_i\in \mathbb{H}^{N_\mathrm{T}},\forall i$, and $\Rank(\mathbf{W}_i)=1,\forall i,$ in (\ref{eqn:SDP}) are imposed
  to guarantee that $\mathbf{W}_i=\mathbf{w}_i\mathbf{w}^H_i$. We note that in constraint C3,
   $\mathbf{\Psi}_n=\mathbf{e}_n\mathbf{e}_n^H$ and $\mathbf{e}_n$ is the
   $n$-th unit vector of length $N_{\mathrm{T}}$ and $\big[\mathbf{e}_n\big]_{t,1}=0,\forall t\ne n$.  The transformed problem above is still non-convex due to the
  rank constraint in C6. However, if we relax this constraint (remove it from the problem formulation), the transformed problem becomes a convex SDP and can be solved efficiently by off-the-shelf numerical solvers such as SeDuMi \cite{JR:SeDumi}. It is known that if the obtained solution $\mathbf{W}_i$ for the relaxed problem admits a rank-one matrix $\forall i$, then it is the optimal solution of the original problem in (\ref{eqn:SDP}). Yet, the proposed constraint relaxation may not be tight, i.e., a rank-one solution may not exist,  and, in this case, the result of the relaxed problem serves as a performance upper bound for the original problem, since a larger feasible solution set is considered in the relaxed problem. In the following, we provide  a sufficient condition for $\Rank(\mathbf{W}_i)=1,\forall i,$ of the relaxed problem and exploit it for the design of two suboptimal power allocation schemes.

 \subsection{Optimality Conditions}
In this subsection, the tightness of the proposed  SDP relaxation is investigated by studying   the Karush-Kuhn-Tucker (KKT) conditions and  the dual problem of the relaxed version of problem (\ref{eqn:SDP}). To this end, we first need
the Lagrangian function  of  (\ref{eqn:SDP}) which is given by
\begin{eqnarray}\hspace*{-2mm}&&\notag{\cal
L}(\mathbf{W}_i,\mathbf{V},\boldsymbol \lambda,\boldsymbol{\beta},\boldsymbol{\delta},\mathbf{Y}_i,\mathbf{Z})\\
\notag\hspace*{-5mm}&=&\hspace*{-3mm} \sum_{i=1}^L \Tr(\mathbf{W}_i)+\Tr(\mathbf{V})-\sum_{i=1}^L\Tr(\mathbf{Y}_i\mathbf{W}_i)-\Tr(\mathbf{ZV}) \\
\notag\hspace*{-5mm}&+&\hspace*{-3mm}\sum_{i=1}^L
\lambda_i\Big(\frac{-\Tr(\mathbf{H}\mathbf{W}_i)}{\Gamma_{\mathrm{req}_i}}+\sum_{j=i+1}^L \Tr(\mathbf{H}\mathbf{W}_j)+\Tr(\mathbf{H}\mathbf{V})+\sigma_{\mathrm{s}}^2\Big)\\
\notag\hspace*{-5mm}&+&\hspace*{-3mm} \sum_{k=1}^{K-1} \beta_k \Big(\frac{\Tr(\mathbf{G}_k\mathbf{W}_1)}{\Gamma_{\mathrm{tol}_k}}\hspace*{-1mm}-\hspace*{-1mm}\sum_{j=2}^L \Tr(\mathbf{G}_k\mathbf{W}_j)-\Tr(\mathbf{G}_k\mathbf{V})-\sigma_{\mathrm{s}}^2\Big)\\
\hspace*{-5mm}&+&\hspace*{-3mm}\sum_{n=1}^{N_{\mathrm{T}}}
 \delta_n \Big(\Tr\Big(\mathbf{\Psi}_n \big(\mathbf{V}+\sum_{i=1}^L\mathbf{W}_i \big)\Big)- P_{\max_n}\Big).
\label{eqn:Lagrangian}
\end{eqnarray}
Here, $\boldsymbol\lambda$, with elements $\lambda_i\ge 0,\forall i\in\{1,\ldots,L\}$, is the Lagrange multiplier vector associated with  the minimum required SINR for decoding layer $i$ for the desired receiver in C1. $\boldsymbol \beta$, with elements $\beta_k\ge 0,\forall k\in\{1,\ldots,K-1\}$, is
the vector of Lagrange multipliers for the maximum tolerable SINRs of the potential eavesdroppers in C2.   $\boldsymbol \delta$,  with elements $\delta_{n}\ge 0,\forall n\in\{1,\ldots,N_{\mathrm{T}}\}$,  is the Lagrange multiplier vector
for the per-antenna maximum transmit power in C3.  Matrices $\mathbf{Z},\mathbf{Y}_i\succeq \mathbf{0}$ are the Lagrange multipliers for the positive semi-definite constraints on matrices $\mathbf{V}$ and $\mathbf{W}_i$  in C4 and C5, respectively.  The dual problem for the SDP relaxed problem is given by
\begin{eqnarray}\label{eqn:dual}
\underset{ \underset{\mathbf{Y}_i,\mathbf{Z}\succeq \mathbf{0}}{\boldsymbol \lambda,\boldsymbol{\beta},\boldsymbol{\delta}\succeq \mathbf{0}}}{\maxo}\ \underset{{\mathbf{W}_i,\mathbf{V}\in \mathbb{H}^{N_\mathrm{T}}}}{\mino}\,\,{\cal
L}(\mathbf{W}_i,\mathbf{V},\boldsymbol \lambda,\boldsymbol{\beta},\boldsymbol{\delta},\mathbf{Y}_i,\mathbf{Z}).\label{eqn:master_problem}
\end{eqnarray}

For facilitating the presentation in the sequel, we define $\mathbf{W}_l^*, \mathbf{V}^*$ as the optimal solution of the relaxed version of problem  (\ref{eqn:SDP}).
In the following proposition, we provide a sufficient condition\footnote{We found by simulations that the SDP relaxation can admit a rank-one solution   in some instances even if the sufficient condition stated in Proposition 1 is not satisfied.} for a rank-one matrix solution for the relaxed version of problem (\ref{eqn:SDP}).

\begin{proposition}For $\Gamma_{\mathrm{req}_i}>0,\forall i,$ in the relaxed version  of problem (\ref{eqn:SDP}), $\Rank(\mathbf{W}_1^*)=1$ always holds.  On the other hand,  $\Rank(\mathbf{W}_j^*)=1,\forall j\in\{2,\ldots,L\},$ holds when $\beta_k=0, \forall k\in\{1,\ldots,K-1\}$.
\end{proposition}
\begin{proof}Please refer to the Appendix.
\end{proof}

In the following, two suboptimal
power allocation schemes are proposed based on Proposition 1.

\subsubsection{Suboptimal Power Allocation Scheme 1}
It can be concluded from Proposition 1 that when constraint C2  is independent of optimization variable $\mathbf{W}_i$, the solution of the relaxed problem has a rank-one structure for $\mathbf{W}_i,\forall i$. Thus, for facilitating an efficient power allocation algorithm design, we replace constraint  C2 in (\ref{eqn:SDP}) by $\mbox{\textoverline{C2}}$ and the new optimization problem is  as follows:
\begin{eqnarray}\label{eqn:suboptimal1}
&& \hspace*{-8mm}\underset{\mathbf{W}_i,\mathbf{V}\in \mathbb{H}^{N_\mathrm{T}}
}{\mino}\,\, \sum_{i=1}^L \Tr(\mathbf{W}_i)+\Tr(\mathbf{V})\\
\notag \mbox{s.t.} &&\hspace*{10mm}\mbox{C1, C3, C4, C5 }\\
&&\hspace*{-2mm}\mbox{\textoverline{C2}: }\notag\frac{\Tr(\mathbf{G}_k\mathbf{W}_1)}
{\Tr(\mathbf{G}_k\mathbf{V})+\sigma_{\mathrm{s}}^2} \le \Gamma_{\mathrm{tol}_k},\forall k\notag .
\end{eqnarray}
 We note that the new constraint $\mbox{\textoverline{C2}}$ does not take into account the contribution of video signal layer $2,\ldots,L$ at the potential eavesdroppers. In particular,
 \begin{eqnarray}
\frac{\Tr(\mathbf{G}_k\mathbf{W}_1)}
{\sum_{j=2}^L \Tr(\mathbf{G}_k\mathbf{W}_j)+\Tr(\mathbf{G}_k\mathbf{V})+\sigma_{\mathrm{s}}^2}\le \frac{\Tr(\mathbf{G}_k\mathbf{W}_1)}
{\Tr(\mathbf{G}_k\mathbf{V})+\sigma_{\mathrm{s}}^2}
 \end{eqnarray}
 holds and  replacing constraint C2 by $\mbox{\textoverline{C2}}$ results in a smaller feasible solution set for the original problem. Thus,  the obtained solution of problem (\ref{eqn:suboptimal1}) serves as a performance lower bound for the original optimization problem (\ref{eqn:SDP}). Besides, it is expected that the problem formulation in (\ref{eqn:suboptimal1}) requires a higher   artificial noise power compared to (\ref{eqn:SDP}), in order to fulfill constraint $\mbox{\textoverline{C2}}$. We note that (\ref{eqn:suboptimal1}) is a convex optimization problem which can be solved by numerical solvers. On the other hand, we can follow a similar approach as in the Appendix to verify that the sufficient condition in Proposition 1 is always satisfied for problem formulation (\ref{eqn:suboptimal1}) and thus a rank-one solution always results.

\subsubsection{Suboptimal Power  Allocation Scheme 2}
  A hybrid scheme is proposed as suboptimal power  allocation scheme 2.  It computes the solutions for suboptimal scheme 1 and  the SDP relaxation in (\ref{eqn:SDP}) in parallel and selects one of the solutions, cf. Table \ref{table:algorithm}. Specifically, when  the upper bound solution of SDP relaxation is not tight, i.e., $\exists i:\Rank(\mathbf{W}_i)>1$, we adopt the solution given by the proposed suboptimal scheme 1. Otherwise,  the proposed scheme 2 will select the solution given by SDP relaxation since the global optimal is achieved when $\Rank(\mathbf{W}_i^*)=1,\forall i$.
  \begin{Remark}
  We note that  suboptimal power allocation scheme 2 is based on the solution of two convex optimization problems with polynomial time computational complexity.
  \end{Remark}

\begin{table}[t]\caption{Suboptimal Power Allocation Scheme}\label{table:algorithm}
\vspace*{-6mm}
\begin{algorithm} [H]           \setcounter{algorithm}{1}          
\floatname{algorithm}{Suboptimal Power Allocation Scheme}
\caption{}          
\label{alg1}                           
\begin{algorithmic} [1]
\normalsize           
\STATE Solve the relaxed version of  problem (\ref{eqn:SDP})
\IF {the solution of the relaxed version of problem (\ref{eqn:SDP}) is rank-one, i.e., $\Rank(\mathbf{W}_i)=1,\forall i$, } \STATE  $\mbox{Global optimal soultion}=\,$\TRUE \RETURN
 ${\mathbf W}_i$, ${\mathbf  V}=$ solution of the relaxed version of problem (\ref{eqn:SDP})
 \ELSE
\STATE Solve  problem (\ref{eqn:suboptimal1}) and $\mbox{Lower bound solution}=\,$ \TRUE\RETURN
 ${\mathbf W}_i$, ${\mathbf  V}=$ solution of problem (\ref{eqn:suboptimal1})
 \ENDIF
\end{algorithmic}

\end{algorithm}
\vspace*{-9mm}
\end{table}

\section{Results}
\label{sect:result-discussion} In this section, we evaluate the
system performance for the proposed resource allocation schemes using simulations.  We consider a single cell
communication system with $K$ receivers and the corresponding simulation parameters are provided in Table \ref{tab:feedback}.
 The system performance is obtained by averaging over 10000 multipath  fading realizations. Given the system parameters in Table \ref{tab:feedback}, in all considered scenarios,  the minimum secrecy capacity
 of layer $1$ video information is bounded  below by $C_{\mathrm{sec}}=\log_2(1+\Gamma_{\mathrm{req}_1})-\log_2(1+\Gamma_{\mathrm{tol}_k})\ge 2.179$ bit/s/Hz.

\subsection{Average Total Transmit Power versus Number of Potential Eavesdroppers }
Figure \ref{fig:p_SNR} depicts the  average total transmit power versus the
 number of idle receivers (potential eavesdroppers), $K-1$, for $N_{\mathrm{T}}=4$ transmit antennas and different power allocation schemes.
It can be observed that  the average total transmit power for the proposed schemes
 increases with the number of potential eavesdroppers.  In fact, the transmitter has to generate a higher amount of artificial noise
 for providing secure communication when there are more potential eavesdroppers in the system. Besides, the two
  proposed suboptimal power allocation schemes perform close to the performance upper bound  achieved by SDP relaxation.
   In particular,  the performance of  proposed scheme 1  serves as a lower bound (i.e., a higher average transmit power)
    for the two proposed schemes since a smaller feasible
    solution set is considered in (\ref{eqn:suboptimal1}).  On the other hand, the performance of proposed scheme 2 coincides with the
    SDP relaxation upper bound. This is due to the fact that proposed scheme 2 is a hybrid scheme which exploits the possibility
    of achieving the global optimal solution
       via SDP relaxation.

For comparison, Figure
\ref{fig:p_SNR} also contains the average total transmit power of two baseline
power allocation schemes. For baseline scheme 1, we adopt single layer transmission for  delivering the video signal. In particular,
we solve the corresponding optimization problem with respect to $\{\mathbf{W}_1,\mathbf{V}\}$ subject to constraints C1--C5 via
semi-definite relaxation. The minimum required SINR   for decoding the single layer video information  at the desired receiver for baseline scheme 1 is set to $\Gamma_{\mathrm{req}}^{\mathrm{Single}}=2^{\sum_{i=1}^L\log_2(1+\Gamma_{\mathrm{req}_i})}-1$. It is expected that the minimum required secrecy capacity in baseline scheme 1 is higher than that in multilayer transmission. In baseline scheme 2,
layered video transmission is considered. Specifically, we adopt maximum ratio transmission (MRT) \cite{JR:TWC_large_antennas} for delivering  the video information of each layer, i.e., we apply a fixed direction for beamforming vector $\mathbf{w}^{\mathrm{sub}}$, where
$\mathbf{w}^{\mathrm{sub}}$ is the eigenvector corresponding to the maximum eigenvalue of $\mathbf{H}$. Then, we set the
 beamforming matrix
as $u_i\mathbf{W}=u_i \mathbf{w}^{\mathrm{sub}}(\mathbf{w}^{\mathrm{sub}} )^H $, where $u_i$ is a new non-negative scalar optimization variable for adjusting the power of $u_i\mathbf{W}$. Subsequently,
 we adopt the same setup as in (\ref{eqn:SDP}) for optimizing $\{u_i,\mathbf{V}\}$ and obtain a suboptimal rank-one solution  $u_i\mathbf{W}_i$.
\begin{table}[t]\caption{System parameters}\label{tab:feedback} \centering\vspace*{-3mm}
\begin{tabular}{|L|L|}\hline
User distance  & $50 \mbox{ m}$  \\
                                                        \hline
Eavesdroppers distance & $30 \mbox{ m}$  \\
    \hline
Multipath fading distribution & \mbox{Rayleigh  fading }  \\
 \hline
Path loss model & $34.53 + 38 \log_{10}(d) $ (dB), $d$ is the receiver distance in meter  \\
    \hline
    Maximum transmit power per-antenna $P_{\max_n}$  &  $43$ dBm   \\
        \hline

 Number of layers $L$ &  3  \\
       \hline
    Minimum requirement on the SINR of layers  $[\Gamma_{\mathrm{req}_1}\,\,\Gamma_{\mathrm{req}_2}\,\,\Gamma_{\mathrm{req}_3}]$ & $[6\,\, 9 \,\,12]$ dB   \\
       \hline
          Maximum tolerable SINR at eavesdropper $\Gamma_{\mathrm{tol}_k},\forall k$ & $-10$ dB   \\
       \hline
        Gaussian noise power $\sigma_{\mathrm{s}}^2$ &  \mbox{$-95$ dBm}   \\
       \hline
\end{tabular}\vspace*{-3mm}
\end{table}
\begin{figure}[t]
 \centering \vspace*{-6mm}
\includegraphics[width=3.5 in]{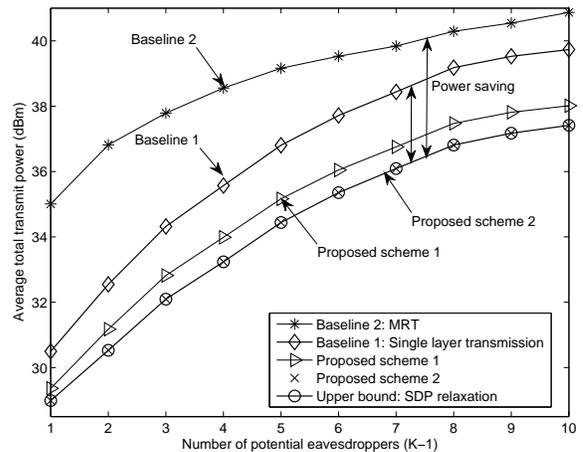}\vspace*{-2mm}
\caption{Average total transmit power (dBm) versus number of idle receivers (potential eavesdroppers), $K-1$, for $N_{\mathrm{T}}=4$ and different power allocation schemes.
The double-sided arrows indicate the power savings achieved by the proposed schemes compared to the baseline schemes.} \label{fig:p_SNR}\vspace*{-3mm}
\end{figure}It can be observed that baseline 1 requires a higher total average power  compared to the proposed power allocation schemes. This is attributed to the fact that single layer transmission for providing secure communication in  baseline 1 does not
 posses the   \emph{self-protecting} structure  that layered transmission has. As a result,
 a higher transmit power for artificial noise generation is required in baseline 1 to ensure secure video delivery. On the other hand,
 although layered transmission is adopted in baseline 2,
 the performance of   baseline 2 is the worst among all the schemes. The reason for this is that the transmitter loses
  degrees of freedom in power allocation
when the structure of video information beamforming matrix, $\mathbf{W}_i,\forall i,$ is fixed to $ \mathbf{w}^{\mathrm{sub}}(\mathbf{w}^{\mathrm{sub}} )^H$ which causes a serious performance degradation.

\begin{figure}[t]
 \centering
\includegraphics[width=3.5 in]{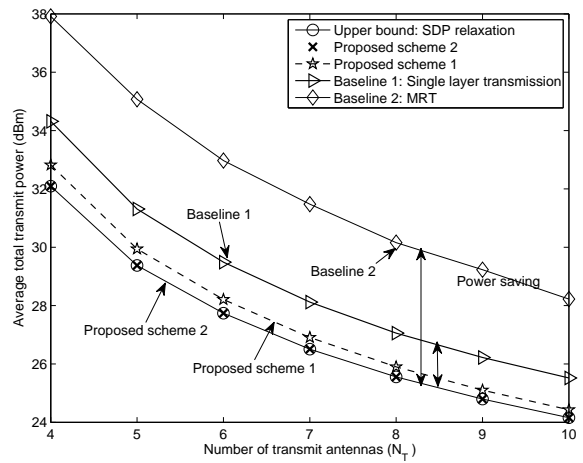}\vspace*{-2mm}
\caption{Average total transmit power (dBm) versus total number of antennas, $N_{\mathrm{T}}$, for $K-1=3$ idle receivers
 (potential eavesdroppers) and different power allocation schemes.
The double-sided arrows indicate the power savings achieved by the proposed schemes compared to the baseline schemes. } \label{fig:pt_nt}\vspace*{-3mm}
\end{figure}

\subsection{Average Total Transmit Power versus Number of Antennas}

Figure \ref{fig:pt_nt} shows the average total transmit power versus the
number of transmit antennas for $K-1=3$ potential eavesdroppers and different power allocation schemes.
 It is expected that the average total transmit power
decreases  with increasing number of antennas for all power allocation schemes. This is because extras degrees of freedom
 can be exploited for power allocation when more antennas are available at the transmitter.
On the other hand,   the proposed schemes  provide  substantial power savings compared to  both baseline schemes for all considered
scenarios due to the adopted layered transmission and the optimization of $\{\mathbf{W}_i,\mathbf{V}\}$.

\section{Conclusions}\label{sect:conclusion}
In this paper,  we focused on  the power allocation
algorithm design for  secure layered video transmission. The algorithm design was
 formulated as a non-convex optimization problem taking into account artificial noise generation to weaken the channel of potential eavesdroppers. Exploiting SDP, a power allocation algorithm was developed to solve the relaxed version of the non-convex optimization problem which resulted in an upper bound solution for minimization of the total transmit power. Subsequently,
 two suboptimal power allocation schemes were designed  by exploiting the structure of the upper bound solution. Simulation results
  unveiled the power savings enabled by layered transmission and the optimization of beamforming and artificial noise generation
    for facilitating secure video transmission. In our future work, we will consider  the impact of imperfect
      channel state information on secure layered transmission systems.

\section*{Appendix - Proof of Proposition 1}
It can be verified that the relaxed version of problem (\ref{eqn:SDP}) is jointly convex with respect to the optimization variables and satisfies Slater's constraint qualification. Thus, the KKT conditions provide the necessary and sufficient conditions \cite{book:convex} for the optimal solution of the relaxed problem. In the following, we study the rank of $\mathbf{W}_i$ by focusing  on the corresponding KKT conditions:
\begin{eqnarray}\label{eqn:KKT}
\hspace*{-3mm}\mathbf{Y}^*_i\hspace*{-3mm}&\succeq&\hspace*{-3mm} \mathbf{0},\quad\beta_k^*,\,\delta_n^*,\,\lambda_i^*,\ge 0, \forall i,n,k,\\
\hspace*{-3mm} \mathbf{Y}^*_i\mathbf{W}^*_i\hspace*{-3mm}&=&\hspace*{-3mm}\mathbf{0},\forall i, \label{eqn:KKT-complementarity}\\
\label{eqn:KKT_Y1}
\hspace*{-3mm}\mathbf{Y}^*_1\hspace*{-3mm}&=&\hspace*{-3mm}\mathbf{\Upsilon}+\hspace*{-0.5mm}  \sum_{k=1}^{K-1} \frac{\mathbf{G}_k}{\Gamma_{\mathrm{tol}_k}} \beta_k^*     -\frac{\lambda^*_1}{\Gamma_{\mathrm{req}_1}}\mathbf{H},\\
\label{eqn:KKT_Yi}
\hspace*{-3mm}\mathbf{Y}^*_i\hspace*{-3mm}&=&\hspace*{-3mm}\mathbf{\Upsilon} -\hspace*{-0.5mm}  \sum_{k=1}^{K-1} \mathbf{G}_k \beta_k^* +\Big( \sum_{j<i} \lambda_j^* \hspace*{-0.5mm}   -\hspace*{-0.5mm}\frac{\lambda^*_i}{\Gamma_{\mathrm{req}_i}}\Big)\mathbf{H}, i>1,
\end{eqnarray}
where $\mathbf{\Upsilon}=\mathbf{I}_{N_\mathrm{T}}+\hspace*{-0.5mm}\sum_{n=1}^{N_{\mathrm{T}}}\delta_n^*\mathbf{\Psi}_n$ and $\mathbf{Y}^*_i,\,\beta_k^*,\,\delta_n^*,\,\lambda_i^*$ are the optimal Lagrange multipliers for (\ref{eqn:dual}). Equation (\ref{eqn:KKT-complementarity}) is the complementary slackness condition.

The proof is divided into two parts. In the first part, we prove that $\Rank(\mathbf{W}_1^*)=1$. To this end, we post-multiply both sides of (\ref{eqn:KKT_Y1}) by $\mathbf{W}^*_1$ and after some manipulations we obtain for
 \begin{eqnarray}\label{eqn:post}
\Big( \mathbf{\Upsilon}\hspace*{-0.5mm} +\hspace*{-0.5mm}  \sum_{k=1}^{K-1} \frac{\mathbf{G}_k}{\Gamma_{\mathrm{tol}_k}} \beta_k^*\Big)\mathbf{W}^*_1 =\frac{\lambda^*_1}{\Gamma_{\mathrm{req}_1}}\mathbf{H}\mathbf{W}^*_1.
 \end{eqnarray}
Since $\mathbf{I}_{N_\mathrm{T}}\hspace*{-0.5mm} +\hspace*{-0.5mm}  \sum_{k=1}^{K-1} \frac{\mathbf{G}_k}{\Gamma_{\mathrm{tol}_k}} \beta_k^* $ is a positive definite matrix,  the following equality holds \cite{book:matrix_analysis}:
 \begin{eqnarray}
&&\Rank\Big(\mathbf{W}^*_1\Big)=\Rank\Big(\big(\mathbf{\Upsilon}\hspace*{-0.5mm} +\hspace*{-0.5mm}  \sum_{k=1}^{K-1} \frac{\mathbf{G}_k}{\Gamma_{\mathrm{tol}_k}} \beta_k^*\big)\mathbf{W}^*_1\Big) \notag\\
=&&\Rank\Big(\frac{\lambda^*_1}{\Gamma_{\mathrm{req}_1}}\mathbf{H}\mathbf{W}^*_1\Big)\notag\\
\label{eqn:rank_option}
=&&\min\Big\{\Rank\Big(\frac{\lambda^*_1}{\Gamma_{\mathrm{req}_1}}\mathbf{H}\Big),\Rank\Big(\mathbf{W}^*_1\Big)\Big\}.
 \end{eqnarray}
We note that $\Rank(\mathbf{H})=1$ and $\Rank(\frac{\lambda^*_1}{\Gamma_{\mathrm{req}_1}}\mathbf{H})$ is either zero or one. On the other hand,
$\mathbf{W}^*_1\ne\mathbf{0}$ is required to satisfy the minimum SINR requirement of the desired receiver  in C1 of (\ref{eqn:SDP}) when $\Gamma_{\mathrm{req}_1}>0$. As a result,  we need to  prove $\lambda^*_1>0$ in order to show $\Rank(\mathbf{W}^*_1)=1$. In other words, constraint C1 has to be satisfied with equality for $i=1$. In the following, we prove by contradiction that constraint C1 is indeed  satisfied with equality $\forall i$, i.e.,
\begin{eqnarray}
\frac{\Tr(\mathbf{H}\mathbf{W}_i)}{\sum_{j=i+1}^L \Tr(\mathbf{H}\mathbf{W}_j)+\Tr(\mathbf{H}
\mathbf{V})+\sigma_{\mathrm{s}}^2} = \Gamma_{\mathrm{req}_i},\forall i.
\end{eqnarray}
 Without loss of generality, we denote $({\mathbf{W}}^*_i,{\mathbf{V}}^*)$ as the optimal solution.  Suppose that for layer $a\in\{1,\ldots,L\}$, $\mbox{C1}$ is
 satisfied with strict inequality, i.e., ``$>$", at the optimal solution. Then, we construct a new feasible solution $\mathbf{\tilde W}_i=\mathbf{W}^*_i$ except in layer $a$: $\mathbf{\tilde W}_a= \alpha\mathbf{W}_a^*$.  $0<\alpha<1$ is
  imposed in the new solution such that the considered constraint is satisfied with equality. It can be verified that the new solution ($\mathbf{\tilde W}_i,{\mathbf{V}}^*)$ achieves a lower objective value in (\ref{eqn:SDP}) than $({\mathbf{W}}^*_i,{\mathbf{V}}^*)$. Thus, $({\mathbf{W}}^*_i,{\mathbf{V}}^*)$ cannot be the optimal solution. As a result, $\lambda^*_i>0,\forall i$, holds for the optimal solution. By combining (\ref{eqn:rank_option}) and $\lambda_1^*>0$, $\Rank(\mathbf{W}^*_1)=1$ holds true at the optimal solution.

For the second part of the proof, we show that  $\Rank(\mathbf{W}^*_j)=1, j\in\{2,\ldots,L\}$, under the assumption of $\beta_k^*=0,\forall k$. We only provide a sketch of the proof due to page limitation.  From (\ref{eqn:KKT_Yi}), we obtain for $\beta^*_k=0,\forall k,$
 \begin{eqnarray}\label{eqn:general}
\Big(\mathbf{\Upsilon}\hspace*{-0.5mm} -\hspace*{-0.5mm}  \sum_{k=1}^{K-1} \mathbf{G}_k \beta_k^* \Big)\mathbf{W}^*_i=\mathbf{\Upsilon}\mathbf{W}^*_i=\Big( \frac{\lambda^*_i}{\Gamma_{\mathrm{req}_i}}\hspace*{-0.5mm} -\hspace*{-0.5mm}\sum_{j<i} \lambda_j^* \Big)\mathbf{H}\mathbf{W}^*_i
 \end{eqnarray}
It can be seen that (\ref{eqn:general}) has a similar structure as (\ref{eqn:post}). Thus,  we can use a similar approach as in the first part
of the proof to show that $\Rank(\mathbf{W}^*_j)=1, j\in\{2,\ldots,L\}$ if $\beta_k^*=0, \forall k$.

\bibliographystyle{IEEEtran}
\bibliography{OFDMA-AF}

\begin{thebibliography}{10}
\providecommand{\url}[1]{#1}
\csname url@samestyle\endcsname
\providecommand{\newblock}{\relax}
\providecommand{\bibinfo}[2]{#2}
\providecommand{\BIBentrySTDinterwordspacing}{\spaceskip=0pt\relax}
\providecommand{\BIBentryALTinterwordstretchfactor}{4}
\providecommand{\BIBentryALTinterwordspacing}{\spaceskip=\fontdimen2\font plus
\BIBentryALTinterwordstretchfactor\fontdimen3\font minus
  \fontdimen4\font\relax}
\providecommand{\BIBforeignlanguage}[2]{{%
\expandafter\ifx\csname l@#1\endcsname\relax
\typeout{** WARNING: IEEEtran.bst: No hyphenation pattern has been}%
\typeout{** loaded for the language `#1'. Using the pattern for}%
\typeout{** the default language instead.}%
\else
\language=\csname l@#1\endcsname
\fi
#2}}
\providecommand{\BIBdecl}{\relax}
\BIBdecl

\bibitem{book:david_wirelss_com}
D.~Tse and P.~Viswanath, \emph{{Fundamentals of Wireless Communication}},
  1st~ed.\hskip 1em plus 0.5em minus 0.4em\relax {Cambridge University Pres},
  2005.

\bibitem{JR:TWC_large_antennas}
D.~W.~K. Ng, E.~Lo, and R.~Schober, ``{Energy-Efficient Resource Allocation in
  OFDMA Systems with Large Numbers of Base Station Antennas},'' \emph{IEEE
  Trans. Wireless Commun.}, vol.~11, pp. 3292 --3304, Sep. 2012.

\bibitem{JR:MIMO_layered}
K.~Bhattad, K.~Narayanan, and G.~Caire, ``{On the Distortion SNR Exponent of
  Some Layered Transmission Schemes},'' \emph{IEEE Trans. Inf. Theory},
  vol.~54, pp. 2943--2958, Jun. 2008.

\bibitem{JR:MIMO_layered2}
D.~Song and C.~W. Chen, ``{Scalable H.264/AVC Video Transmission Over MIMO
  Wireless Systems With Adaptive Channel Selection Based on Partial Channel
  Information},'' \emph{IEEE Trans. Circuits Syst. Video Technol.}, vol.~17,
  pp. 1218--1226, Sep. 2007.

\bibitem{CN:MISO_layer}
J.~Xuan, S.~H. Lee, and S.~Vishwanath, ``{Broadcast Strategies for MISO and
  Multiple Access Channels},'' in \emph{Proc. IEEE Personal, Indoor and Mobile
  Radio Commun. Sympos.}, 2007, pp. 1--4.

\bibitem{JR:Video_layers}
Y.~Fallah, H.~Mansour, S.~Khan, P.~Nasiopoulos, and H.~Alnuweiri, ``{A Link
  Adaptation Scheme for Efficient Transmission of H.264 Scalable Video Over
  Multirate WLANs},'' \emph{IEEE Trans. Circuits Syst. Video Technol.},
  vol.~18, pp. 875--887, Jun. 2008.

\bibitem{JR:Video_layers2}
B.~Barmada, M.~Ghandi, E.~Jones, and M.~Ghanbari, ``{Prioritized Transmission
  of Data Partitioned H.264 Video with Hierarchical QAM},'' \emph{IEEE Signal
  Process. Lett.}, vol.~12, pp. 577--580, Jul. 2005.

\bibitem{Report:Wire_tap}
A.~D. Wyner, ``{The Wire-Tap Channel},'' Tech. Rep., Oct 1975.

\bibitem{CN:Multicast}
Q.~Li and W.-K. Ma, ``{Multicast Secrecy Rate Maximization for MISO Channels
  with Multiple Multi-Antenna Eavesdroppers},'' in \emph{Proc. IEEE Intern.
  Commun. Conf.}, 2011, pp. 1--5.

\bibitem{JR:Artifical_Noise1}
S.~Goel and R.~Negi, ``{Guaranteeing Secrecy using Artificial Noise},''
  \emph{IEEE Trans. Wireless Commun.}, vol.~7, pp. 2180 -- 2189, Jun. 2008.

\bibitem{JR:Kwan_physical_layer}
D.~W.~K. Ng, E.~S. Lo, and R.~Schober, ``{Secure Resource Allocation and
  Scheduling for OFDMA Decode-and-Forward Relay Networks},'' \emph{IEEE Trans.
  Wireless Commun.}, vol.~10, pp. 3528--3540, Aug. 2011.

\bibitem{JR:secrecy_beamforming}
X.~Chen and L.~Lei, ``{Energy-Efficient Optimization for Physical Layer
  Security in Multi-Antenna Downlink Networks with QoS Guarantee},'' \emph{IEEE
  Commun. Lett.}, vol.~17, pp. 637--640, Mar. 2013.

\bibitem{JR:interference_model}
A.~Iyer, C.~Rosenberg, and A.~Karnik, ``{What is the Right Model for Wireless
  Channel Interference?}'' \emph{IEEE Trans. Wireless Commun.}, vol.~8, pp.
  2662--2671, May 2009.

\bibitem{JR:SeDumi}
J.~F. Sturm, ``{Using SeDuMi 1.02, A MATLAB Toolbox for Optimization over
  Symmetric Cones},'' \emph{{Optimiz. Methods and Software}}, vol. 11-12, pp.
  625--653, Sep. 1999.

\bibitem{book:convex}
S.~Boyd and L.~Vandenberghe, \emph{{Convex Optimization}}.\hskip 1em plus 0.5em
  minus 0.4em\relax {Cambridge University Press}, 2004.

\bibitem{book:matrix_analysis}
R.~A. Horn and C.~R. Johnson, \emph{{Matrix Analysis}}.\hskip 1em plus 0.5em
  minus 0.4em\relax {Cambridge University Press}, 1985.

\end{thebibliography}

\end{document}